\documentclass[journal]{IEEEtran}
\usepackage{setspace,cite,graphics,amsmath,rotating}

\newcommand{\abs}[1]{\left| #1 \right|}
\newtheorem{theorem}{\textbf{Theorem}}

\begin{document}
\title{Distributed Power Allocation Strategies for Parallel Relay Networks}

\author{Min~Chen,~
        Semih~Serbetli,~\IEEEmembership{Member,~IEEE,}
        and~Aylin~Yener,~\IEEEmembership{Member,~IEEE}
\thanks{Manuscript received December 7, 2005; revised February 19, 2007; accepted October 25, 2007.
The editor coordinating the review of this paper and approving it
for publication was Kin K. Leung. This work was supported in part by
NSF grants CCF 02-37727, CNS 05-08114, CNS 06-26905 and DARPA
ITMANET Program grant W911NF-07-1-0028. This work was presented in
part in Globecom 2005, St. Louis, MO, November 2005.}
\thanks{Min Chen and Aylin Yener are with the Wireless
Communications and Networking Laboratory, Department of Electrical
Engineering, Pennsylvania State University, University Park, PA
16802. email: mchen@psu.edu, yener@ee.psu.edu.}
\thanks{Semih Serbetli was with the
Wireless Communications and Networking Laboratory at the
Pennsylvania State University. He is now with Philips Research Labs,
Eindhoven, Netherlands. email: semih.serbetli@philips.com.} }

\markboth{IEEE Transactions on Wireless Communications, accepted for
publication}{}
%


\maketitle

\begin{abstract}
We consider a source-destination pair assisted by parallel
regenerative decode-and-forward relays operating in orthogonal
channels. We investigate distributed power allocation strategies for
this system with limited channel state information at the source and
the relay nodes. We first propose a distributed decision mechanism
for each relay to individually make its decision on whether to
forward the source data. The decision mechanism calls for each relay
that is able to decode the information from the source to compare
its relay-to-destination channel gain with a given threshold. We
identify the optimum distributed power allocation strategy that
minimizes the total transmit power while providing a target
signal-to-noise ratio at the destination with a target outage
probability. The strategy dictates the optimum choices for the
source power as well as the threshold value at the relays. Next, we
consider two simpler distributed power allocation strategies, namely
the \emph{passive source model} where the source power and the relay
threshold are fixed, and the \emph{single relay model} where only
one relay is allowed to forward the source data. These models are
motivated by limitations on the available channel state information
as well as ease of implementation as compared to the optimum
distributed strategy. Simulation results are presented to
demonstrate the performance of the proposed distributed power
allocation schemes. Specifically, we observe significant power
savings with proposed methods as compared to random relay selection.
\end{abstract}

\begin{IEEEkeywords}
Relay selection, distributed power allocation, decode-and-forward,
orthogonal parallel relays.
\end{IEEEkeywords}

\section{Introduction}
Relay-assisted transmission schemes for wireless networks are
continuing to flourish due to their potential of providing the
benefits of space diversity without the need for physical antenna
arrays \cite{cover_cap}. Among the earliest work on cooperative
networks are references \cite{Erkip_user1,Erkip_user2,Laneman_coop}.
A cooperative diversity model is proposed in \cite{Erkip_user1} and
\cite{Erkip_user2}, in which two users act as partners and
cooperatively communicate with a common destination, each
transmitting its own bit in the first time interval and the
estimated bit of its partner in the second time interval. In
\cite{Laneman_coop}, several low-complexity cooperative protocols
are proposed and studied, including fixed relaying, selection
relaying and incremental relaying, in which the relay node can
either amplify-and-forward (AF) or decode-and-forward (DF) the
signal it receives. In \cite{Laneman_space}, networks consisting of
more than two users that employ the space-time coding to achieve the
cooperative diversity are considered. Coded cooperation schemes are
discussed in \cite{Erkip_coding} and \cite{Jan-Nosratinia}, where a
user transmits part of its partner's codeword as well. References
\cite{Gupta_towards} and \cite{Kramer_cooperative} investigate the
capacity of relay networks of arbitrary size. References so far have
shown that, relay nodes can provide performance improvement in terms
of outage behavior \cite{Laneman_coop,Laneman_space}, achievable
rate
region\cite{Erkip_user1,Erkip_user2,Gupta_towards,Kramer_cooperative},
and error probability
\cite{Erkip_coding,Jan-Nosratinia,Ribeiro_symbol, Anghel_exact}.

Power efficiency is a critical design consideration for wireless
networks such as ad-hoc and sensor networks, due to the limited
transmission power of the (relay and the source) nodes. To that end,
choosing the appropriate relays to forward the source data, as well
as the transmit power levels of all the nodes become important
design issues. Optimum power allocation strategies for relay
networks are studied up-to-date for several structures and relay
transmission schemes. Three-node models are discussed in
\cite{Madsen_capacity} and \cite{Brown_resource}, while multi-hop
relay networks are studied in
\cite{Reznik_degraded,Hasna_optimal,Dohler_resource}. Relay
forwarding strategies for both AF and DF parallel relay channels in
wideband regime are proposed in \cite{Maric_forwarding}. Recent
works also discuss relay selection algorithms for networks with
multiple relays. Optimum relay selection strategies for several
models are identified in
\cite{Ribeiro_symbol,Maric_forwarding,Cai_achievable}. Recently
proposed practical relay selection strategies include pre-select one
relay \cite{Luo_link}, best-select relay \cite{Luo_link},
blind-selection-algorithm \cite{Lin_relay},
informed-selection-algorithm \cite{Lin_relay}, and cooperative relay
selection \cite{Zheng_effectiveness}. All of these proposed methods
result in power efficient transmission strategies. However, the
common theme is that, the implementations of these algorithms
require either the destination or the source to have substantial
information about the network, such as the channel state information
(CSI) of all communication channels, received signal-to-noise ratio
(SNR) at every node, the topology of the network, etc. Such
centralized power allocation/relay selection schemes may be
infeasible to implement due to the substantial feedback
requirements, overhead and delay they may introduce.

To overcome the obstacles of a centralized architecture, several
heuristic approaches have been proposed in
\cite{Hunter_distributed}, for multi-user networks with coded
cooperation. In this work, users select cooperation partners based
on a priority list in a distributed manner. Although the proposed
algorithms are advantageous due to their ease of implementation,
their performance depends on the fading conditions, and the
randomness in the channel may prevent the protocols from providing
full diversity. In \cite{Herhold}, an SNR threshold method is
proposed for the relay node to make a decision on whether to forward
the source data in a three-node model. Since there is only one relay
node in the considered system, relay selection is not an issue.
Reference \cite{Bletsas_simple} provides a relay selection algorithm
based on instantaneous channel measurements done by each relay node
locally. For the purpose of reducing the communication among relays,
a flag packet is broadcasted by the selected relay to notify the
other relays of the result.

In this paper, we investigate \emph{optimum distributed power
allocation} strategies for decode-and-forward parallel relay
networks, in which only partial CSI is accessible at the source and
the relay nodes. We first propose a distributed decision mechanism
for each relay node to individually make a decision on whether to
forward the source data. In contrast to the SNR based decision
protocol presented in \cite{Herhold}, in our proposed decision
mechanism, the relay makes its decision not only by considering its
received SNR, but also by comparing its relay-to-destination channel
gain with a given threshold, and no feedback from the destination is
needed. The overall overhead is further reduced as compared to the
method proposed in \cite{Bletsas_simple} since the distributed
decision mechanism does not require communication among relays.
Secondly, given such a relay decision scheme, and considering an
outage occurs whenever the SNR at the destination is lower than the
required value (target), we formulate the distributed power
allocation problem that aims to minimize the expected value of the
total transmit power while providing the target SNR at the
destination with an outage probability constraint. We identify the
solution of this problem, that consists of the optimum value of the
source power, and the corresponding relay decision threshold based
on the partial CSI available at the source. The extra power the
distributed power allocation mechanism needs as compared to the
optimum centralized power allocation mechanism, i.e., the
\emph{additional power expenditure}, is examined to observe the
tradeoff between the outage probability and the additional power
expenditure.

We next consider two special cases with simpler implementation,
namely the \emph{passive source model} where the source does not
contribute to the relay selection process, and the \emph{single
relay model} where one relay node is selected to forward the source
data based on limited CSI. For each case, we optimize the respective
relevant parameters. Our results demonstrate that considerable power
savings can be obtained by our proposed distributed relay selection
and power allocation schemes with respect to random relay selection.

The organization of the paper is as follows. In
Section~\ref{sysmodel}, the system model is described. The
distributed power allocation problem is formulated and the optimum
solution is given in Section~\ref{DRA}. In Section~\ref{SS}, we
investigate the passive source model and the single relay model.
Numerical results supporting the theoretical analysis are presented
in Section~\ref{NR}, and Section~\ref{Conc} concludes the paper.

\section{System Model and Background}
\label{sysmodel}

We consider a relay network consisting of a source-destination pair
and $N$ relay nodes employing decode-and-forward. We assume that the
relay nodes operate in pre-assigned orthogonal channels, e.g. in
non-overlapping time/frequency slots, or using orthogonal
signatures. The source is assumed to transmit in a time slot prior
to (and non-overlapping with) the relays. Let $f_i$ and $g_i$ denote
the fading coefficients of the source-to-relay and
relay-to-destination channels for the $i\mbox{th}$ relay node, for
$i=1,...,N$. The fading coefficient of the source-to-destination
link is denoted by $h$. We assume that each channel is flat fading,
and $f_i$, $g_i$ and $h$ are all independent realizations of zero
mean complex Gaussian random variables with variances
$\sigma_{f_i}^2$, $\sigma_{g_i}^2$ and $\sigma_h^2$ per dimension,
respectively.

\begin{figure} [t]
\centering
\includegraphics[width=3.5in]{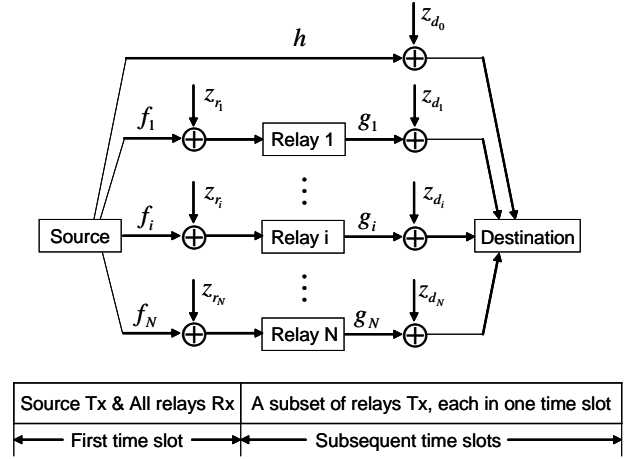}
\caption{Relay network system model.} \label{fig-sysmodel}
\end{figure}
Without loss of generality, we will assume that we have a time
slotted system in the sequel. The system model is shown in Figure
\ref{fig-sysmodel}. In the first time slot, the source broadcasts
$X_o$ with power $P_s$. The destination observes $y_{d_0}$:
\begin{equation}
y_{d_0}=\sqrt{P_s}hX_o+z_{d_0}
\end{equation}
and the $i$th relay observes $y_{r_i}$:
\begin{equation}
y_{r_i}=\sqrt{P_s}f_iX_o+z_{r_i}~~~\mbox{for}~~i=1,...,N
\end{equation}
where $z_{d_0}$ and $\{z_{r_i}\}_{i=1}^N$ are Additive White
Gaussian Noise (AWGN) terms at the destination and the relays,
respectively. Assume without loss of generality that they are of
variance $1/2$ per dimension. The $i$th relay node is said to be
{\it reliable} and can correctly decode $X_o$ when its received SNR,
$SNR_{r_i}$, satisfies
\begin{equation}
\label{reliable} SNR_{r_i}=P_s\abs{f_i}^2\geq{SNR_{target}}
\end{equation}
where ${SNR_{target}}$ is the given decodability constraint. In the
subsequent time slots following the first one, the relays that
belong to the set of reliable relays, $A_R$, can decode and forward
the source data to the destination, each in its assigned time slot.
Throughout this paper, we assume that the reliable relays simply
regenerate the source data $X_o$
\cite{Laneman_coop,Brown_resource,Hasna_optimal}. The signal
received at the destination from the reliable relay $i$ is
\begin{equation}
y_{d_i}=\sqrt{P_i}g_iX_o+z_{d_i}, ~~~~~~ i\in A_R
\end{equation}
where $P_i$ is the transmit power of the $i$th relay node, and
$z_{d_i}$ is the AWGN term at the $i$th relay-to-destination
channel. The destination combines signals received from the reliable
relay nodes and the direct link with a maximum ratio combiner (MRC),
and the resulting SNR at the destination is
\begin{equation}
{SNR}_d=P_s\abs{h}^2+\sum_{i \in A_R} P_i\abs{g_i}^2
\end{equation}
We consider that the destination can correctly receive the source
data whenever ${SNR_{d}}\geq{SNR_{target}}$.

Given this system model, the power allocation problem for
regenerative DF relay networks with parallel relays can be posed as
\begin{eqnarray}
\label{orig}
\underset{Ps, \{P_i\}} \min  & P_s+\sum_{i \in A_R}P_i\\
 \text{s. t.  }&  P_s\abs{h}^2+\sum_{i \in
A_R}P_i\abs{g_i}^2\geq{SNR_{target}} \label{qosatsource}\\
& P_s\abs{f_i}^2\geq{SNR_{target}} ~~\text{
 for each}~ i \in A_R
\end{eqnarray}
We note that the resulting power allocation strategy may prevent
some reliable relays from participating simply by assigning zero
power to those relays.

The optimum power allocation strategy for DF relay networks using
different code books at the relays is identified in
\cite{Maric_forwarding}. This strategy, re-stated below for the
benefit of the reader, is easily seen to be the optimum {\it
centralized} power allocation strategy for regenerative DF relay
networks as well.
\begin{eqnarray}
P_s^* &=& \frac{{SNR_{target}}}{\abs{f_{k^*}}^2} \label{sourceoptpower}\\
P_{i}^* &=& \left\{ \begin{array}{ll} \left(
\frac{{SNR_{target}}-\abs{h}^2{SNR_{target}}/\abs{f_{k^*}}^2}
{\abs{g_{k^*}}^2}\right)^+,~i=k^*\\
0,\text{otherwise}
\end{array} \right. \label{relayoptpower}\\
k^* &=& \arg\underset{\{k \in
A_E\}}\min\left[\frac{1}{\abs{f_k}^2}+\frac{1}{\abs{g_k}^2}-\frac{\abs{h}^2}{\abs{f_k}^2\abs{g_k}^2}\right]\label{optrelsel}
\end{eqnarray}
where $(\cdot)^+=\max(0,\cdot)$. In (\ref{optrelsel}), the set $A_E$
denotes the set of efficient relays such that the transmission
through the relay is more power efficient than the direct
transmission, i.e.,
\begin{equation}
A_E={\{i|(\abs{f_i}^2\ge|h|^2) \cap (\abs{g_i}^2\ge|h|^2), i=1..N\}}
\end{equation}

Observe that when the source power is assigned as in
(\ref{sourceoptpower}), the relay node  $k^*$, chosen according to
(\ref{optrelsel}), is the only relay node with received SNR equal to
$SNR_{target}$. Thus, each relay node can decide whether it is the
intended relay node by simply checking its received SNR. When the
SNR contribution of the relay node,
${{SNR_{target}}-\abs{h}^2{SNR_{target}}/\abs{f_{k^*}}^2}$, is
indicated explicitly by the source, the intended relay node can
calculate its required transmit power as in (\ref{relayoptpower})
and forward $X_o$ to the destination. Alternatively, the source can
broadcast the selected relay and the optimum power level in a side
channel.

A moment's thought reveals that to implement the strategy given by
(\ref{sourceoptpower})-(\ref{optrelsel}), the full CSI, i.e.,
$\{f_i,g_i\}_{i=1}^N$ and $h$, at the source node, and the
individual CSI, i.e., $\{f_i,g_i\}$, at relay node $i$ are needed.
Although (\ref{sourceoptpower})-(\ref{optrelsel}) provides the most
power efficient DF relay transmission strategy, its centralized
nature, i.e., the fact that it requires the channel estimate of each
link and the feedback of this information to the source, may render
its implementation impractical. As such, distributed strategies are
needed. In the following, we devise efficient distributed power
allocation strategies.

\section{Distributed Power Allocation}
\label{DRA} Our aim in this paper is to find power allocation
schemes that {\it do not require a centralized mechanism}, and
utilize the limited available CSI at each node. In practice, it is
feasible that the channels are estimated by training before the
actual data transmission, when each node operates in TDMA mode. When
the source transmits the training bits, all relay nodes can
simultaneously estimate their source-to-relay fading coefficients
$\{f_i\}_{i=1}^N$ due to the broadcast nature of the wireless
medium. Similarly, when the relay node $i$ transmits the training
bits, the source-to-relay coefficient $f_i$ can be estimated at the
source. However, for $\{g_i\}_{i=1}^N$ to be available at the
source, the feedback from the destination for each realization is
required, which may be impractical. Thus, we investigate distributed
power allocation schemes when the source has the realizations
$\{f_i\}_{i=1}^N$ and $h$, and only the statistics of $\{g_i\}$. The
relay nodes are assumed to have their individual CSI, i.e., $f_i$
and $g_i$ for relay $i$, $i=1,...,N$.

\subsection{Distributed Decision Mechanism}
\label{DDM} We first derive a distributed decision mechanism with
the model assumptions given above. Since the source has only the
statistical description instead of the realizations
$\{g_i\}_{i=1}^N$, the optimum centralized power allocation
indicated by (\ref{sourceoptpower})-(\ref{optrelsel}) cannot be
implemented by the source. Also, while it is clear that for a fixed
source power, the best strategy is transmitting through the reliable
relay node that has the highest relay-to-destination channel gain,
this mechanism requires a comparison of all $\{g_i\}_{i=1}^N$. The
distributed nature of the strategy requires that each relay should
make its decision relying only on its individual CSI. Since each
relay can easily determine whether it is a reliable relay by using
its SNR value, i.e., its individual CSI, we propose that the $i$th
reliable relay decides it will be a forwarding node when its channel
gain to the destination satisfies
\begin{eqnarray}
\label{best} \abs{g_i}^2\geq\gamma
\end{eqnarray}
where $\gamma$ is a given threshold value. Relay $i$ then forwards
the decoded signal with {\it sufficient} power. That is, we have
\begin{eqnarray}
P_{i}^{\ast}= SNR'_{target}/{\abs{g_i}^2}
\end{eqnarray}
where $SNR'_{target}=(SNR_{target}-P_s |h|^2)^+$ denotes the SNR
contribution from the relay.\footnote{$\gamma$ and $SNR'_{target}$
values are assumed to be broadcasted by the source on a side
channel.}

We note that such a distributed decision mechanism includes the
probability that more than one relay will transmit. Similarly, we
note that with any $\gamma>0$, the scheme results in a nonzero
probability that none of the relay nodes satisfies (\ref{best}), and
hence a nonzero outage probability
$\text{Prob}(SNR_d<SNR_{target})$. As such, the source should
determine the optimum source power and the corresponding threshold
$\gamma$ by considering the realizations of $\{f_i\}$ and the
randomness in $\{g_i\}$, to meet a system given specification, i.e.,
an outage probability requirement.

\subsection{Source Power Allocation and Threshold Decision}
Given the above described strategy, 
we now investigate how the source should decide the value of its
transmit power $P_s$ and the relay decision threshold $\gamma$, to
satisfy the target SNR, $SNR_{target}$ at the destination with a
target outage probability, $\rho_{target}$.

From the source's point of view, the relay transmit powers are
random variables with known statistics because the realizations
$\{g_i\}_{i=1}^N$ are not available at the source. We have the pdf
of $X_i=|\mathrm{g_i}|^2$ as
\begin{equation}
p_{X_i}(x_i)=\frac{1}{2\sigma^2_{g_i}}\exp\left(-\frac{x_i}{2\sigma^2_{g_i}}\right),
~~\text{for}~~ i \in [1,...,N]
\end{equation}
where $\mathrm{g_i}$ is a zero mean complex Gaussian random variable
with variance $\sigma^2_{g_i}$ per dimension. We consider the
expected value of the transmit power of relay $i$
\begin{eqnarray}
\label{relaypower} E{[P_i]}&=&\int_\gamma^\infty \frac{SNR'_{target}}{x_i}p_{X_i}(x_i) dx_i\\
&=&\int_\gamma^\infty \frac{SNR'_{target}}{2\sigma_{g_i}^2x_i}
\exp(-\frac{x_i}{2\sigma_{g_i}^2})dx_i
\end{eqnarray}
The distributed power allocation problem can then be expressed as
\begin{eqnarray}
\label{optprob}
\underset{\gamma,P_s}\min & P_s+\sum_{i \in A_R(P_s)}E{[P_i]}\\
\text{s. t.}& \text{Prob}(SNR_{d} \le {SNR_{target}}) \le \rho_{target} \label{probconstraint}\\
& ~~~P_s\abs{f_i}^2\geq{SNR_{target}}~~\text{
 for each}~i \in A_R
\end{eqnarray}
where we explicitly state the dependency of the set of reliable
relay $A_R$ on $P_s$. Observe that the deterministic
quality-of-service guarantee in (\ref{qosatsource}) is replaced by
the probabilistic constraint (\ref{probconstraint}). The following
theorem provides the optimum solution:
\begin{theorem}\label{theorem1}The optimum source power, $P_s^{**}$, can only be one of the $(M+1)$ discrete values in the set
\begin{eqnarray}
\{\frac{SNR_{target}}{|f_1|^2},...,\frac{SNR_{target}}{|f_M|^2},
\frac{SNR_{target}}{|h|^2}\} \label{set}
\end{eqnarray}
where we reorder the indices of the relay nodes such that $|f_1|^2
> |f_2|^2 > ...> |f_M|^2 > |h|^2> |f_{M+1}|^2... >
|f_N|^2$, i.e.,
$\frac{SNR_{target}}{|f_1|^2}<\frac{SNR_{target}}{|f_2|^2}<...
<\frac{SNR_{target}}{|f_M|^2}<\frac{SNR_{target}}{|h|^2}<\frac{SNR_{target}}{|f_{M+1}|^2}<...<\frac{SNR_{target}}{|f_N|^2}$.
\footnote{$P_s=\frac{SNR_{target}}{|h|^2}$ is the largest candidate
of the source power. With this power level, source can reach the
destination via the direct link and relay transmission is not
needed.} For each possible $P_s^{**}$ value, there exist a
corresponding reliable relay set $A_R^{**}$, and a unique optimum
threshold value, $\gamma^{**}$.
\end{theorem}
\begin{proof} Assume that $P_s=SNR_{target}/|f_i|^2$ and there
exist a reliable relay set $A_R^\dag$ containing $R_i$ relay nodes
and a corresponding threshold value $\gamma^\dag$. Then, the
expected value of the total power is
\begin{eqnarray}\label{ptotal}
E[P_{total}]=P_s+\nonumber\qquad\qquad\qquad\qquad\qquad\qquad\qquad\qquad\\
\sum_{i \in A_R^\dag} \int_{\gamma^\dag}^\infty
\frac{(SNR_{target}-P_s |h|^2)^+}{2\sigma_{g_i}^2x_i}
\exp(\frac{-x_i}{2\sigma_{g_i}^2})dx_i
\end{eqnarray}
We consider the set of transmitting relays as a super relay node
whose effective channel gain to the destination is
$\abs{g_{eff}}^2$. Thus, the expected value of the total power can
be expressed as
\begin{equation}
E[P_{total}]=P_s+\frac{(SNR_{target}-P_s |h|^2)^+}{\abs{g_{eff}}^2}
\end{equation}
where
\begin{equation}
\abs{g_{eff}}^2=\frac{1}{\sum_{i \in A_R^\dag}
\int_{\gamma^\dag}^\infty \frac{{1}}{2\sigma_{g_i}^2x_i}
\exp(-x_i/2\sigma_{g_i}^2) dx_i}
\end{equation}

The direct transmission is more power efficient than the
relay-assisted transmission when the channel gain of the direct
link, $\abs{h}^2$, is greater than the effective channel gain of the
relay-to-destination links, $\abs{g_{eff}}^2$, i.e.,
\begin{eqnarray}
|h|^2>\abs{g_{eff}}^2
\end{eqnarray}
In this case, the optimum source power is
$P_s^{**}=SNR_{target}/|h|^2$.

On the other hand, the relay transmission is preferred when
\begin{eqnarray}\label{relayisprefered}
|h|^2<\abs{g_{eff}}^2
\end{eqnarray}
We note that the derivative of $E[P_{total}]$ with respect to $P_s$
is
\begin{eqnarray}
\frac{\partial E[P_{total}]}{\partial P_s}=1 -
\frac{|h|^2}{\abs{g_{eff}}^2}
\end{eqnarray}
and (\ref{relayisprefered}) implies $\frac{\partial
E[P_{total}]}{\partial P_s}>0$, which means increasing $P_s$ beyond
$ SNR_{target}/|f_i|^2$ until the value $SNR_{target}/|f_{i+1}|^2$
for $i=1,\ldots,M$ does not change $A_R^\dag$ but increases the
expected value of the total power $E[P_{total}]$. Thus, the optimum
source power $P_s^{**}$ can be only one of the (M+1) discrete values
in the set given by (\ref{set}).

For $P_s=SNR_{target}/|f_i|^2$, one of the candidates of the optimum
source power, and its corresponding reliable set $A_R^\dag$, when
$\gamma $ increases, the expected value of the total power
decreases, while the outage probability increases. Therefore,
threshold $\gamma^\dag$ should be chosen as the value that satisfies
the outage probability with equality, i.e.,
\begin{eqnarray} \label{threshold}
\prod _{i\in A_R^\dag} (1-\int_{\gamma^\dag}^\infty
\frac{{1}}{2\sigma_{g_i}^2} \exp(-\frac{x_i}{2\sigma_{g_i}^2})
dx_i)=\rho_{target}
\end{eqnarray}
It can be further reduced to
\begin{equation}
\prod _{i\in
A_R^\dag}(1-\exp(-\frac{\gamma^\dag}{2\sigma_{g_i}^2}))=\rho_{target}
\end{equation}
Let $\sigma_{g_{min}}^2$$=$$\min\{\sigma_{g_i}^2, i\in A_R^\dag\}$
and $\sigma_{g_{max}}^2$$=$$\max\{\sigma_{g_i}^2, i\in A_R^\dag\}$,
we have
\begin{eqnarray}
(1-\exp(-\frac{\gamma^\dag}{2\sigma_{g_{max}}^2}))^{|A_R^\dag|}\leq
\prod_{i\in
A_R^\dag}(1-\exp(-\frac{\gamma^\dag}{2\sigma_{g_i}^2}))\nonumber\\
\leq(1-\exp(-\frac{\gamma^\dag}{2\sigma_{g_{min}}^2}))^{|A_R^\dag|}
\end{eqnarray}
Therefore, $\gamma^\dag$ is bounded as
\begin{equation}
\gamma^\dag_{min}\leq \gamma^\dag \leq \gamma^\dag_{max}
\end{equation}
where
$\gamma^\dag_{min}=-\ln(1-(\rho_{target})^{\frac{1}{|A_R^\dag|}})\cdot2\sigma_{g_{min}}^2$
and
$\gamma^\dag_{max}=-\ln(1-(\rho_{target})^{\frac{1}{|A_R^\dag|}})\cdot2\sigma_{g_{max}}^2$.
The value of $\gamma^\dag$ can be obtained by a search in the range
$[\gamma^\dag_{min}, \gamma^\dag_{max}]$ numerically.

Note that for $P_s=SNR_{target}/|h|^2$, i.e., when the source can
reach the destination via the direct link, $\gamma^\dag=\infty$ to
prevent any redundant relay transmission and power consumption.
\end{proof}

The source should simply compare $(M+1)$ possible $E[P_{total}]$
values and decide the best $(P_s^{**}, \gamma^{**})$ pair. Note that
when the expected value of the total transmit power is higher than
that with direct transmission, the source will prefer to transmit
directly to the source.\footnote{The source would communicate this
decision via the side channel.}

The cost of the lack of full CSI at the source, i.e., the cost of
using the distributed relay decision mechanism, is an additional
power expenditure. Let $P_{total}^{**}$ and $P^*_{total}$ denote the
total power of the proposed optimum distributed power allocation
scheme, and that of the optimum centralized allocation scheme which
is the sum of the source power $P_s^*$ and the relay power $P_i^*$
given in (\ref{sourceoptpower})-(\ref{optrelsel}), respectively. The
expected value of the additional power expenditure is:
\begin{equation}\label{waste}
E[P_{add}]=E[P^{**}_{total}]-E[P^*_{total}]\qquad\qquad\qquad\qquad\qquad\qquad\qquad\qquad\qquad
\end{equation}
\begin{equation}
 \label{waste1}= P_s^{**} +\sum_{i \in A_R^{**}}
\int_{\gamma^{**}}^\infty \frac{SNR'_{target}}{2\sigma_{g_i}^2x}
\exp(-x/2\sigma_{g_i}^2) dx-E[P^*_{total}]
\end{equation}

We observe that in (\ref{threshold}), $\rho_{target}$ is an
increasing function of $\gamma$, while in (\ref{waste1}),
$E[{P_{add}}]$ is a decreasing function of $\gamma$. Thus, there
exists a tradeoff between the outage probability and the additional
power expenditure: reducing the target outage probability will
require more additional power. While designing the power allocation
strategy, a reasonable target outage probability should be chosen in
accordance with this tradeoff.

\section{Simpler Schemes}\label{SS}
The optimum distributed power allocation strategy still requires the
realizations of $\{f_i\}_{i=1}^N$ and $h$, i.e., the CSI of the
source-to-relay and the direct links, available at the source. It
also requires the source to update the threshold $\gamma^{**}$ and
the source power $P_s^{**}$ at each time when these channel
coefficients change. Due to further limitations on the availability
of this CSI and for implementation complexity, we may opt for even
simpler schemes. In this context, we next consider two special
cases, namely the passive source model and the single relay model.
For both cases, we have the previous assumption that each relay has
its individual CSI, i.e., $f_i$ and $g_i$ for relay $i$,
$i=1,...,N$. Below are the brief descriptions of the two models.
\begin{itemize} \item {\it {Passive source model:}} We assume that the
source only has the statistics of {\it all} communication channels,
and does not participate in the relay selection process at all. For
this model, we fix the source power $P_s$, and the relay decision
threshold $\gamma$, and employ the same distributed decision
mechanism as proposed in \ref{DDM}.
\end{itemize}
\begin{itemize} \item {\it {Single relay
model:}} We assume that the source has CSI of the direct and the
source-to-relay links, i.e., $\{f_i\}_{i=1}^N$ and $h$, and the
statistics of the relay-to-destination links $\{g_i\}$. We have the
source select {\it one} assisting relay node to satisfy the system
requirements on received SNR and the outage probability.
\end{itemize}

\subsection{Passive Source Model} \label{PSM} In practice, we may have situations where the
source does not have the realizations of any of the channels, but
has access only to the statistical descriptions of them. It may also
be the case that the source may not be able to do computationally
expensive operations, e.g., due to hardware constraints in sensor or
RFID networks. We term such source nodes, {\it passive}. Considering
these practical issues, in this section, we investigate the
distributed power allocation for the passive source model.

Since each relay has its individual CSI, we can apply the same
distributed decision mechanism as proposed in Section \ref{DDM}.
However, a passive source cannot optimize its power $P_s$ or
$\gamma$ based on channel realizations; they should be found
off-line based on the statistical descriptions of the channel and
kept fixed for all realizations. Note that, different from Section
\ref{DRA}, in this case, we may end up having no reliable relay if
the fixed source power value is too small.

Let us now develop the criterion on how to choose the source power
$P_s$ and the threshold $\gamma$ by considering the outage
probability and the additional power expenditure jointly. The outage
probability of the direct link is given by
\begin{eqnarray}\label{eq-directlink}
d_{out}=\text{Prob}\{P_s\abs{{\mathrm{h}}}^2<
SNR_{target}\}\nonumber\\
=1-\exp\left(-\frac{{SNR_{target}}}{P_s\cdot
2\sigma_{h}^2}\quad\right)
\end{eqnarray}
For clarity of exposition, let us define $a_i$ as the probability
that the $i$th relay is a reliable relay, $b_i$ as the probability
that the $i$th relay satisfies (\ref{best}), and $c_i$ as the
probability that the $i$th relay is in set $A_C$, which denotes the
set of relays that satisfy both (\ref{reliable}) and (\ref{best}).
We have
\begin{eqnarray}
a_i=\text{Prob}\{i\in
A_R\}=\text{Prob}\{P_s\abs{\mathrm{f_i}}^2\geq{SNR_{target}}\}\nonumber\\
=\exp\left(-\frac{ {SNR_{target}}}{P_s\cdot
2\sigma_{f_i}^2}\right)\qquad\qquad\quad\qquad\qquad\quad
\end{eqnarray}
\begin{equation}
b_i=\text{Prob}\{\abs{\mathrm{g_i}}^2\geq\gamma\}=\exp\left(-\frac{\gamma}{2\sigma_{g_i}^2}\right)\qquad\qquad\quad
\end{equation}
\begin{equation}
c_i=\text{Prob}\{i\in A_C\}=a_i\cdot
b_i\qquad\qquad\quad\qquad\quad\quad
\end{equation}
where $\mathrm{f_i}$ and $\mathrm{g_i}$ are zero mean complex
Gaussian random variables with variances $\sigma^2_{f_i}$ and
$\sigma^2_{g_i}$ per dimension, respectively. The overall outage
probability becomes
\begin{eqnarray}
\label{overallout} \rho_{outage}=\text{Prob}\{A_C=\emptyset\}\cdot
d_{out}\qquad\qquad\qquad\qquad\quad
\nonumber\\=\prod_{i=1}^{N}\text{Prob}\{i\not\in A_C\}d_{out}
=\prod_{i=1}^{N}\left[1-c_i\right]\cdot d_{out}
\end{eqnarray}

Observe in (\ref{overallout}) that $\rho_{outage}$ is a function of
the source transmit power, $P_s$ and the threshold $\gamma$. To
choose the $(P_s,\gamma)$ pair that satisfies (\ref{overallout}), we
make two observations. The first one is
\begin{eqnarray}
\label{prop1} \rho_{outage}\ge \prod_{i=1}^{N}\left[1-a_i\right]
d_{out}
\end{eqnarray}
where equality occurs when $\gamma=0$, that is when {\it all}
reliable relays forward the source data. Thus, to achieve a target
outage probability, $\rho_{target}$, there exists a minimum source
power $P_s$, that provides the target outage probability with
$\gamma=0$. Note that when $P_s$ is chosen close to this minimum
value, the corresponding $\gamma$ factor will be close to 0,
resulting in many relays transmitting. This may result in
unnecessarily large extra power expenditure and care must be
exercised to choose the correct pair. Secondly, we observe
\begin{eqnarray}
\label{prop5} \rho_{outage} \ge \prod_{i=1}^{N}\left[1-b_i\right]
d_{out}
\end{eqnarray}
Thus, for a given $P_s$ value, $\gamma$ should be strictly less than
some threshold to provide a target outage probability.

When we consider a special case where $d_{out}\approx 1$, i.e., the
direct link is not reliable, and $\{\mathrm{f_i}\}_{i=1}^N$ and
$\{\mathrm{g_i}\}_{i=1}^N$ are i.i.d., we have
\begin{eqnarray}
\label{prop2} \rho_{outage}\approx (1-\exp{(-\frac{SNR_{target}}{2
P_s \sigma^2_f}-\frac{\gamma}{2\sigma^2_g})})^N
\end{eqnarray}
and $(P_s,\gamma)$ pair that aims to achieve an outage probability
$\rho_{target}$ should satisfy
\begin{eqnarray}
\label{prop3} \frac{SNR_{target}}{2 P_s
\sigma^2_f}+\frac{\gamma}{2\sigma^2_g}\approx
-ln(1-(\rho_{target})^{1/N})
\end{eqnarray}

Since the relays employ the distributed decision mechanism proposed
in \ref{DDM}, there exists a nonzero probability that additional
relay nodes besides the best relay decide to forward the source
data. In this case, additional power is expended. For a realization
of $|\mathrm{g_i}|^2$, $x_i=|g_i|^2\ge \gamma$, the probability that
relay $i$ makes a forwarding decision even though it is not the best
relay in set $A_R$, $W_i(x_i)$, can be expressed as
\begin{eqnarray}
W_i(x_i)\nonumber\qquad\qquad\qquad\qquad\qquad\qquad\qquad\qquad\qquad\qquad\\
=\text{Prob}(\text{Wrong forwarding decision by relay
$i$}|x_i\ge\gamma)
\end{eqnarray}
\begin{eqnarray}
=\text{Prob}\{(i \in A_R)\cap \nonumber\qquad\qquad\qquad\qquad\qquad\qquad\qquad\\
(\exists j \in A_R \text{ and } j \neq i, \text{ such that
}X_j>x_i\ge\gamma)\}
\end{eqnarray}
\begin{eqnarray}
=\text{Prob}\{i \in A_R\}\cdot\nonumber\qquad\qquad\qquad\qquad\qquad\qquad\qquad\quad \\
\text{Prob}\{\exists j \in A_R \text{ and } j \neq i, \text{ such
that }X_j>x_i\ge\gamma\}
\end{eqnarray}
\begin{eqnarray}
=\text{Prob}\{i \in A_R\}\cdot (1-\text{Prob}\{\forall j \in
[1,...,N] \text{ and } j
\neq i, \nonumber\\
(j \notin A_R)\cup((j \in A_R)\cap(X_j<x_i))\})
\end{eqnarray}
\begin{eqnarray}
=\text{Prob}\{i \in A_R\}\cdot(1-\prod_{j=1, i\neq
j}^N(\text{Prob}\{j \notin A_R\}+\nonumber\qquad\quad\\
\text{Prob}\{j \in A_R\}\cdot\text{Prob}\{ X_j<x_i\}))
\end{eqnarray}
\begin{eqnarray}
=a_i\cdot(1-\prod_{j=1, i\neq
j}^N((1-a_j)+a_j\cdot(1-\exp(-\frac{x_i}{2\sigma_{g_j}^2}))))
\end{eqnarray}
\begin{eqnarray}
=\exp(-\frac{SNR_{target}}{P_s2\sigma_{f_i}^2})\cdot\nonumber\qquad\qquad\qquad\qquad\qquad\qquad\quad\\
(1-\prod_{j=1, i\neq j}^N
(1-\exp(-\frac{SNR_{target}}{P_s2\sigma_{f_j}^2})\cdot\exp(-\frac{x_i}{2\sigma_{g_j}^2})
))
\end{eqnarray}

If relay $i$ makes a wrong forwarding decision, it will transmit
with power value $SNR'_{target}/x_i$. In essence, the power of relay
$i$ is wasted, since the relay with the highest relay-to-destination
channel gain in $A_R$ also transmits the source data to the
destination reliably but with a lower power. We have the expected
value of the wasted power of relay $i$, $E[P_{{waste}_i}]$ as
\begin{eqnarray}
\label{psvwaste1} E[P_{{waste}_i}]=\int_{\gamma}^{\infty} W_{i}(x_i)
\frac{{SNR'_{target}}}{x_i}p_{X_i}(x_i)dx_i
\end{eqnarray}
The expected value of the additional power expenditure of all relays
is\footnote{Observe that $E[P_{waste_i}]=0$ if $i$ is an unreliable
relay or the best reliable relay.}
\begin{equation}
\label{psvwaste2} E[P_{{add_{Relay}}}]=\sum_{i=1}^N E[P_{{waste}_i}]
\end{equation}

Observe that in (\ref{overallout}), $\rho_{outage}$ is an increasing
function of $\gamma$ when other parameters are fixed, while in
(\ref{psvwaste2}), the expected value of the additional power
expenditure is a decreasing function of $\gamma$. There exists a
tradeoff between the outage probability and the additional relay
power expenditure. A reasonable pair of the source power and the
threshold $\gamma$ should be chosen by considering both the tradeoff
and the properties of the $(P_s, \gamma)$ pair in
(\ref{overallout}), (\ref{prop1}) and (\ref{prop5}).

\subsection{Single Relay Model} \label{RSM}

The distributed power allocation schemes proposed up to this point
in general result in multiple relays transmitting to the
destination, causing additional power expenditure. In this section,
we investigate the case where only one relay node selected by the
source is allowed to transmit. In contrast to the centralized
solution in (\ref{sourceoptpower})-(\ref{optrelsel}), however, we
consider that the source has limited CSI. In particular, we
re-emphasize that, only the statistical descriptions of the
relay-to-destination channels are available at the source. Adopting
the single relay model, we will see that the task of finding the
threshold value for the relay forwarding decisions can be
substantially simplified as compared to the optimum distributed
strategy.

When relay $k$ is selected, the source transmits with just enough
power $P_s=SNR_{target}/|f_k|^2$ to make relay $k$ a reliable relay.
So, the source-to-relay link does not have outage. However, since
relay $k$ will forward the decoded source data only when its channel
gain to the destination satisfies $|g_k|^2\geq \tau_k$, we may have
an outage on the relay-to-destination link. Observe that, if relay
$k$ decides to forward the data it will do so with power
$P_k=SNR'_{target}/{\abs{g_{k}}^2}$.

Therefore, to satisfy the outage constraint $\rho_{target}$, the
relay-to-destination gain threshold, $\tau_k$ should satisfy
\begin{eqnarray}
\int_{\tau_k}^\infty p_{X_k}(x_k)d(x_k)&=&\int_{\tau_k}^\infty
\frac{1}{2\sigma_{g_k}^2}\exp\left(-\frac{x_k}{2\sigma_{g_k}^2}\right)dx_k \nonumber\\
&=&1-\rho_{target}
\end{eqnarray}
Thus, we have
\begin{equation}
\label{tau} \tau_k=-2\sigma_{g_k}^2\ln(1-\rho_{target})
\end{equation}
The expected value of the transmit power of the relay node is
\begin{eqnarray}
E[P_k]& = & \int_{\tau_k}^\infty \frac{SNR_{target}^{'}}{x_k}
p_{X_k}(x_k) dx_k
\\&=&\frac{\int_{\tau}^\infty
\frac{SNR_{target}^{'}}{x_k} \exp(-x_k/2)
dx_k}{{2\sigma_{g_{k}}^2}}\\
\label{powerk}&=&\frac{SNR_{target}^{'}
K(\tau)}{{2\sigma_{g_{k}}^2}}
\end{eqnarray}
where
\begin{equation}
K(\tau)=\int_\tau^\infty \frac{1}{x_k}\exp(-x_k/2) dx_k
\end{equation}
and $\tau=-2\ln(1-\rho_{target})$. We observe that $E[P_k]$
inversely proportional to the variance of the fading coefficient,
${\sigma_{g_k}^2}$.

The optimum power allocation problem in this case becomes
\begin{eqnarray}
\underset{P_s,k}\min &   P_s+E{[P_k]}\label{singlemodelopt}\\
\text{s. t.}& \text{Prob}(SNR_d \le {SNR_{target}}) \le \rho_{target}\label{singlemodeloutage}\\
&P_s\abs{f_k}^2\geq{SNR_{target}}\label{singlemodeldecode}
\end{eqnarray}
Theorem \ref{theorem1} is valid for
(\ref{singlemodelopt})-(\ref{singlemodeldecode}) as well, i.e.,  the
optimum source power $P_s^{**}$, has to be one of the $(M+1)$
possibilities. The proof follows the same steps with the total power
expression (\ref{optprob}) replaced by (\ref{singlemodelopt}), i.e.,
$\sum_{i \in A_R^\dag} \int_{\gamma^\dag}^\infty
\frac{1}{2\sigma_{g_i}^2x_i} \exp(\frac{-x_i}{2\sigma_{g_i}^2})dx_i$
should be replaced by $\frac{ K(\tau)}{{2\sigma_{g_{k}}^2}}$.

The optimum solution can be expressed as
\begin{equation}
P_s^{**} = SNR_{target}/\abs{f_{k^{**}}}^2 \label{singleoptpower}
\end{equation}
\begin{equation}
k^{**}  = \arg\underset{|h|^2<2\sigma_{g_{k}}^2/
K(\tau)}\min\frac{1}{\abs{f_k}^2}+\frac{K(\tau)}{{2\sigma_{g_k}^2}}\left(1-\frac{|h|^2}{|f_k|^2}
\right)^+ \label{singleoptk}
\end{equation}
(\ref{singleoptpower})-(\ref{singleoptk}) result in only the relay
selected by the source, $k^{**}$, satisfying SNR target. Thus, each
relay can decide whether it is the selected node by examining its
own received SNR.

From (\ref{tau}) and (\ref{powerk}), we note the tradeoff between
the outage probability and the additional power expenditure in this
scheme as well. We also note that the relay threshold $\tau_k$ is a
scaled version of $\sigma_{g_k}^2$ for each relay $i$. The
complexity for calculating the relay threshold at the source is thus
significantly less compared to that of the optimum distributed power
allocation scheme derived in Section \ref{DRA}, making the model and
the corresponding strategy given in this section attractive from a
practical stand point. However, we note that, with this scheme,
since {\it exactly} one relay will be reliable, additional power may
be needed as compared to the optimum distributed strategy to satisfy
the same outage requirement.

\section{Numerical Results}
\label{NR}
\begin{figure} [t]
\centering
\includegraphics[width=3.7in]{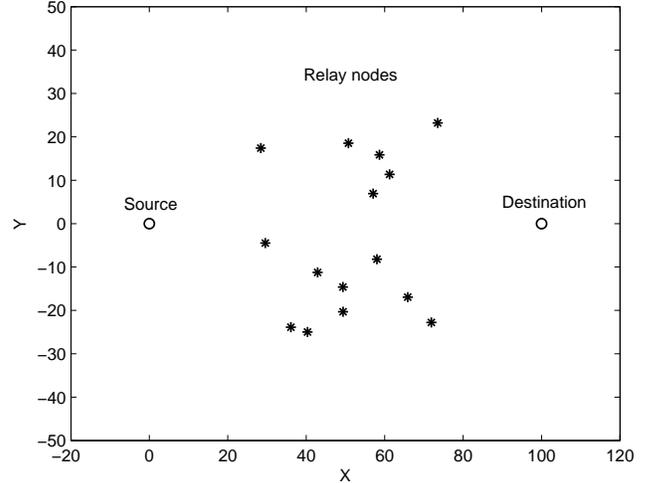}
\caption{System set-up for the simulation.} \label{fig-model}
\end{figure}
In this section, we present numerical results related to the
performance of the proposed distributed power allocation schemes. We
consider a relay network consisting of a source and a destination
$100~m$ apart, and $N=15$ relay nodes that are distributed in a
$50\times50 ~m^2$ square area, as shown in Figure \ref{fig-model}.
We consider the fading model as in \cite{Laneman_coop}, i.e., the
variance of the channel gain is proportional to the distance between
nodes. Thus, we have $\sigma^2_{f_i}=C/d_{SR_i}^{\alpha}$,
 $\sigma^2_{g_i}=C/d_{R_iD}^{\alpha}$ and
$\sigma^2_{h}=C/d_{SD}^{\alpha}$, where $d_{AB}$ is the distance
between node $A$ and $B$, and $S$, $D$ and $R_i$ denote the source,
the destination and the $i$th relay node, respectively. The
path-loss exponent is denoted by $\alpha$. $C$ is a constant that is
expressed as $C=G_tG_r\lambda^2/(4\pi)^2L$, where $G_t$ is the
transmitter antenna gain, $G_r$ is the receiver antenna gain,
$\lambda$ is the wavelength, and $L$ is the system loss factor not
related to propagation ($L \geq 1$). The values $\alpha=3$,
$G_t=G_r=1$, $\lambda=1/3~m$ (carrier frequency $f=900~MHz$), $L=1$,
are used throughout the simulations. The AWGN variances on all
communication links are assumed to be $10^{-10}$. We set
$SNR_{target}=10$ as the system SNR requirement.

Simulation results are presented to demonstrate the performance of
the proposed power allocation strategies. Specifically, we plot
$E[P_{total}]$, the expected value of the total power expended
versus $\rho_{outage}$, the target outage probability. Note that in
the theoretical analysis, there is no outage in the optimum
centralized power allocation (OCPA) in Section \ref{sysmodel}, since
the source and the relay can always adjust their transmit power to
satisfy the SNR requirement at the destination. For a fair
comparison, we define that an outage occurs for OCPA when the total
transmit power is higher than a given power constraint. This is
reasonable since if there is no maximum power constraint, the
expected value of the transmit power goes infinite to achieve a zero
outage probability on a fading channel.

We first compare the performance among the proposed optimum
distributed power allocation (ODPA) scheme, the OCPA scheme, and the
random relay selection (RRS) scheme, in which the source randomly
selects one out of all relays with equal probability to forward the
source data. We observe in Figure \ref{fig-comp} that a substantial
amount of power is saved by employing ODPA, with respect to RRS. The
power savings is more pronounced for low outage probability values.
As expected, an additional power expenditure, which is the penalty
of lack of full CSI, is introduced by ODPA. We observe that the
additional power expenditure decreases as the outage probability
increases, which is expected from the discussion on the tradeoff
between the outage probability and the additional power expenditure
in Section \ref{DRA}.

\begin{figure}[t]
\centering
\includegraphics[width=3.7in]{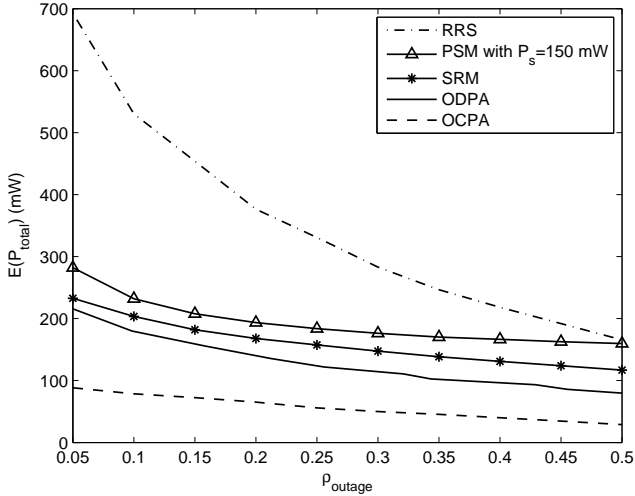}
\caption{$E[P_{total}]$ vs $\rho_{outage}$ for different power
allocation schemes.} \label{fig-comp}
\end{figure}
\begin{figure}[t]
\centering
\includegraphics[width=3.7in]{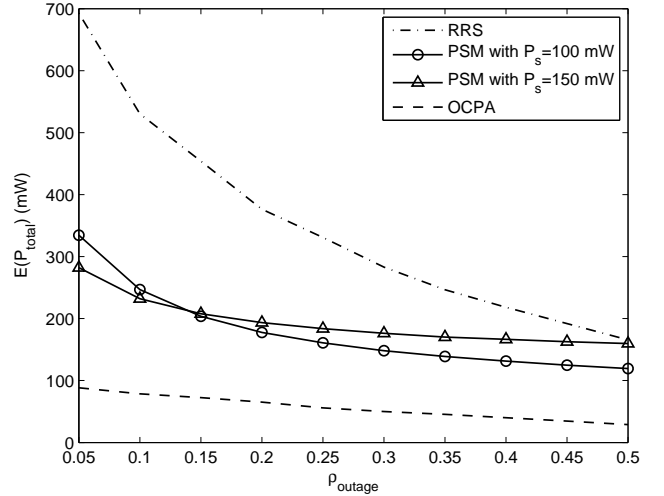}
\caption{$E[P_{total}]$ vs $\rho_{outage}$ for the passive source
model (PSM).} \label{fig-psvsource}
\end{figure}

We also compare all of the proposed distributed power allocation
schemes in Figure \ref{fig-comp}. As expected, we observe that the
best performing scheme is ODPA. Passive source model (PSM) and
single relay model (SRM) both have some performance loss due to the
fact that, for PSM the values of the source power $P_s$ and the
threshold $\gamma$ are fixed; for SRM only one relay node is used
for forwarding transmission. However, the two special cases still
outperform RRS by considering the limited available CSI for power
allocation, and they simplify the optimization process of ODPA and
facilitate the implementations. Thus, PSM and SRM may be preferred
when computational complexity is at a premium. When
$\rho_{outage}=0.05$, approximately, $80\%$, $77\%$ and $67\%$ power
is saved by ODPA, SRM and PSM with respect to RRS, respectively.

Figure \ref{fig-psvsource} remarks that the performance of the
system with PSM depends strongly on the value of the source power
(which is fixed). For low outage probability values, a high source
power is favorable since it reduces the SNR contribution from the
relay nodes, and hence the transmit power of the relay nodes. On the
other hand, for high outage probability values, the source power
becomes a lower bound for the total power. Thus, a low source power
is preferred in this case.

We also investigate the effect of the direct link on the
performance. Figure \ref{fig-psvsccontrbt} and Figure
\ref{fig-sglrlycontrbt} show the effect of the direct link SNR
contribution on PSM and SRM, respectively. It is observed that a
small amount of power savings is obtained when the direct link is
considered. This amount vanishes as the quality of the direct link
decreases. With this observation, when the direct link has a poor
channel quality, the transmitting relay $i$ can forward the signal
with power $SNR_{target}/\abs{g_i}^2$ instead of
$SNR'_{target}/\abs{g_i}^2$ without a significant performance loss.
Employing such a strategy has the advantage that, the direct link,
$h$, is not required for calculating $SNR'_{target}$, and thus the
amount of feedback from the destination is reduced.
\begin{figure}[t]
\centering
\includegraphics[width=3.7in]{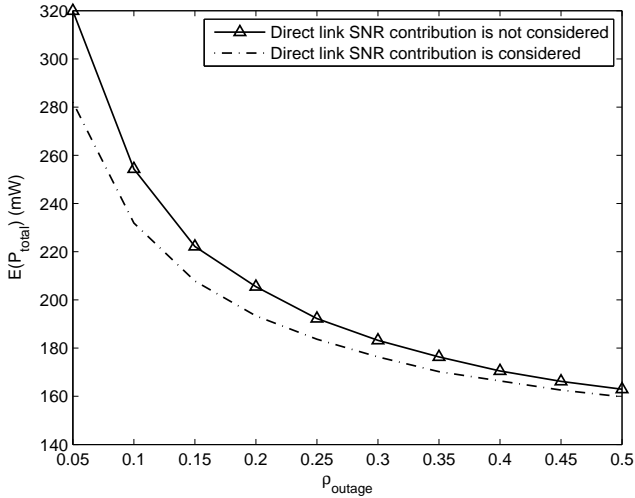}
\caption{Effect of the direct link SNR contribution
 on the passive source model (PSM) ($P_s=150~mW$).} \label{fig-psvsccontrbt}
\end{figure}
\begin{figure}[h]
\centering
\includegraphics[width=3.7in]{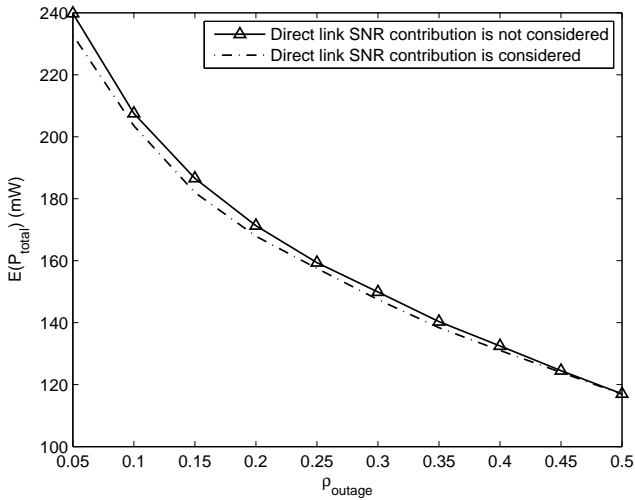}
\caption{Effect of the direct link SNR contribution
 on the single relay model (SRM).} \label{fig-sglrlycontrbt}
\end{figure}
\begin{figure}[h]
\centering
\includegraphics[width=3.7in]{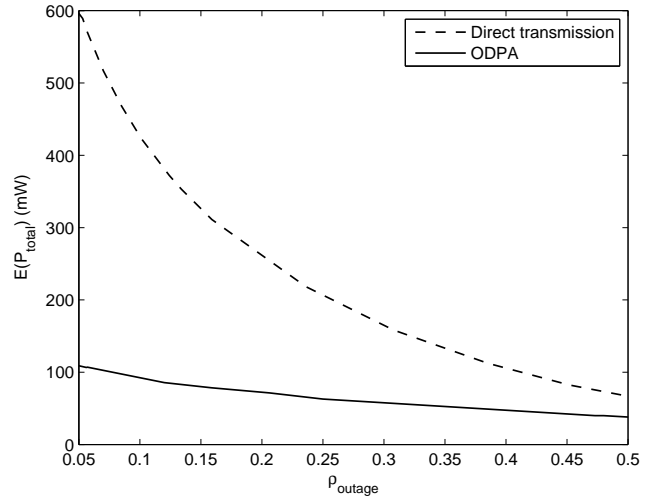}
\caption{Comparison of the relay-assisted transmission scheme ODPA
and the direct transmission scheme.} \label{fig-direct}
\end{figure}

In addition, to show the power efficiency advantage of the
relay-assisted transmission scheme ODPA, we compare the performances
of ODPA and the direct transmission scheme where the signal is
transmitted from the source to the destination via the direct link
only. To show that ODPA benefits more general networks than the one
we considered in Figure \ref{fig-model} where the direct link
distance is larger than that of any source-to-relay or
relay-to-destination link, we now consider that the destination's
position is randomly chosen in the area of $X\times Y=[20,100]\times
[-50,50]$ for each realization, while the source and relay nodes
remain in the same position as in Figure \ref{fig-model}. In Figure
\ref{fig-direct}, we plot the expected value of total power
expenditure, $E[P_{total}]$, versus the target outage probability,
$\rho_{outage}$, for ODPA and the direct transmission scheme. We
observe that in the absence of the relays, direct transmission
scheme requires a relatively high power expenditure to achieve the
same outage probability as compared to ODPA. It is observed that the
proposed relay-assisted transmission scheme provides significant
performance gain in terms of power efficiency upon the direct
transmission. This is intuitively pleasing since the relay selection
and power allocation algorithms in the proposed scheme guarantee
that the more power efficient way is always selected out of the
relay-assisted transmission and the direct transmission for each
channel realization.

\section{Conclusion} \label{Conc}
In this paper, we addressed the distributed power allocation problem
for parallel relay networks. Given the partial CSI available at the
source and the relay nodes, we proposed a distributed relay decision
mechanism and developed the optimum distributed power allocation
scheme. By optimizing the relay selection strategy and power
allocation, the optimum distributed power allocation strategy
performs close to the optimum centralized scheme. We have also
considered two simple distributed power allocation strategies, the
passive source model and the single relay model. Both schemes have
significantly less computational complexity requirements at the
source with a modest sacrifice in performance. Our main result is
that by using distributed power allocation and partial CSI, we can
develop power efficient transmission schemes, reducing the amount of
control traffic overhead for relay-assisted communications.


\begin{thebibliography}{10}

\bibitem{cover_cap}
T.~M. Cover and A.~A.~El Gamal.
\newblock Capacity theorems for the relay channel.
\newblock {\em IEEE Transactions on Information Theory}, IT-25(5):572 -- 584,
  September 1979.

\bibitem{Erkip_user1}
A.~Sendonaris, E.~Erkip, and B.~Aazhang.
\newblock User cooperation diversity-part {I}: System description.
\newblock {\em IEEE Transactions on Communications}, 51(11):1927 -- 1938,
  November 2003.

\bibitem{Erkip_user2}
A.~Sendonaris, E.~Erkip, and B.~Aazhang.
\newblock User cooperation diversity-part {II}: Implementation aspects and
  performance analysis.
\newblock {\em IEEE Transactions on Communications}, 51(11):1939 -- 1948,
  November 2003.

\bibitem{Laneman_coop}
J.~N. Laneman, D.~N.~C. Tse, and G.~W. Wornell.
\newblock Cooperative diversity in wireless networks: Efficient protocols and
  outage behavior.
\newblock {\em IEEE Transactions on Information Theory}, 50(12):3062 -- 3080,
  December 2004.

\bibitem{Laneman_space}
J.~N. Laneman and G.~W. Wornell.
\newblock Distributed space-time-coded protocols for exploiting cooperative
  diversity in wireless networks.
\newblock {\em IEEE Transactions on Information Theory}, 49(10):2415 -- 2425,
  October 2003.

\bibitem{Erkip_coding}
A.~Stefanov and E.~Erkip.
\newblock Cooperative coding for wireless networks.
\newblock {\em IEEE Transactions on Communications}, 52(9):1470 -- 1476,
  September 2004.

\bibitem{Jan-Nosratinia}
M.~Janani, A.~Hedayat, T.~Hunter, and A.~Nosratinia.
\newblock Coded cooperation in wireless communications: Space-time transmission
  and iterative decoding.
\newblock {\em IEEE Transactions on Signal Processing}, 52(2):362 -- 371,
  February 2004.

\bibitem{Gupta_towards}
P.~Gupta and P.~R. Kumar.
\newblock Towards an information theory of large networks: an achievable rate
  region.
\newblock {\em IEEE Transactions on Information Theory}, 49(8):1877 -- 1894,
  August 2003.

\bibitem{Kramer_cooperative}
G.~Kramer, M.~Gastpar, and P.~Gupta.
\newblock Cooperative strategies and capacity theorems for relay networks.
\newblock {\em IEEE Transactions on Information Theory}, 51(9):3037 -- 3063,
  September 2005.

\bibitem{Ribeiro_symbol}
A.~Ribeiro, X.~Cai, and G.~B. Giannakis.
\newblock Symbol error probabilities for general cooperative links.
\newblock {\em IEEE Transactions on Wireless Communications}, 4(4):1264 --
  1273, May 2005.

\bibitem{Anghel_exact}
P.~A. Anghel and M.~Kaveh.
\newblock Exact symbol error probability of a cooperative network in a
  {R}ayleigh-fading environment.
\newblock {\em IEEE Transactions on Wireless Communications}, 3(5):1416 --
  1421, September 2004.

\bibitem{Madsen_capacity}
A.~Host{-}Madsen and J.~Zhang.
\newblock Capacity bounds and power allocation for wireless relay channels.
\newblock {\em IEEE Transactions on Information Theory}, 51(6):2020 -- 2040,
  June 2005.

\bibitem{Brown_resource}
D.~R. {Brown III}.
\newblock Resource allocation for cooperative transmission in wireless
  networks.
\newblock In {\em 38th Asilomar Conference on Signals, Systems and Computers},
  November 2004.

\bibitem{Reznik_degraded}
A.~Reznik, S.~R. Kulkarni, and S.~Verd\'u.
\newblock Degraded {G}aussian multirelay channel: capacity and optimal power
  allocation.
\newblock {\em IEEE Transactions on Information Theory}, 50(12):3037 -- 3046,
  December 2004.

\bibitem{Hasna_optimal}
M.~O. Hasna and M.~{-}S. Alouini.
\newblock Optimal power allocation for relayed transmissions over
  {R}ayleigh-fading channels.
\newblock {\em IEEE Transactions on Wireless Communications}, 3(6):1999 --
  2004, November 2004.

\bibitem{Dohler_resource}
M.~Dohler, A.~Gkelias, and H.~Aghvami.
\newblock Resource allocation for {FDMA}-based regenerative multihop links.
\newblock {\em IEEE Transactions on Wireless Communications}, 3(6):1989 --
  1993, November 2004.

\bibitem{Maric_forwarding}
I.~Maric and R.~Yates.
\newblock Forwarding strategies for {G}aussian parallel-relay networks.
\newblock In {\em Conference on Information Sciences and Systems, CISS'04},
  March 2004.

\bibitem{Cai_achievable}
X.~Cai, Y.~Yao, and G.~B. Giannakis.
\newblock Achievable rates in low-power relay links over fading channels.
\newblock {\em IEEE Transactions on Communications}, 53(1):184 -- 194, January
  2005.

\bibitem{Luo_link}
J.~Luo, R.~S. Blum, L.~J. Cimini, L.~J. Greenstein, and A.~M.
Haimovich.
\newblock Link-failure probabilities for practical cooperative relay networks.
\newblock In {\em IEEE 61st Vehicular Technology Conference, VTC'05 Spring},
  May 2005.

\bibitem{Lin_relay}
Z.~Lin and E.~Erkip.
\newblock Relay search algorithms for coded cooperative systems.
\newblock In {\em IEEE Global Telecommunications Conference 2005, Globecom'05},
  November 2005.

\bibitem{Zheng_effectiveness}
H.~Zheng, Y.~Zhu, C.~Shen, and X.~Wang.
\newblock On the effectiveness of cooperative diversity in ad hoc networks: a
  {MAC} layer study.
\newblock In {\em IEEE International Conference on Acoustics, Speech, and
  Signal Processing, ICASSP'05}, March 2005.

\bibitem{Hunter_distributed}
T.~E. Hunter and A.~Nosratinia.
\newblock Distributed protocols for user cooperation in multi-user wireless
  networks.
\newblock In {\em IEEE Global Telecommunications Conference 2004, Globecom'04},
  December 2004.

\bibitem{Herhold}
P.~Herhold, E.~Zimmermann, and G.~Fettweis.
\newblock A simple cooperative extension to wireless relaying.
\newblock In {\em 2004 International Zurich Seminar on Communications}, pages
  36 -- 39, 2004.

\bibitem{Bletsas_simple}
A.~Bletsas, A.~Khisti, D.~P. Reed, and A.~Lippman.
\newblock A simple cooperative diversity method based on network path
  selection.
\newblock {\em IEEE Journal on Selected Areas of Communications}, 24(3):659 --
  672, March 2006.

\end{thebibliography}
\end{document}